\documentclass[letterpaper, 10 pt, conference]{ieeeconf}  % Comment this line out if you need a4paper

\IEEEoverridecommandlockouts                              % This command is only needed if 
\usepackage{amsmath,amssymb,amsfonts}
\usepackage{algorithmic}
\usepackage{graphicx}
\usepackage{textcomp}
\usepackage{xcolor}

\usepackage{graphics}
\usepackage{graphicx}
\usepackage{epsfig,wrapfig}
\usepackage{color}
\usepackage{float}
\usepackage{calrsfs}
\usepackage{verbatim}

\usepackage{wrapfig}

\newcommand{\R}{\mathbb{R}}
\newcommand{\BM}{\begin{bmatrix}}
\newcommand{\EM}{\end{bmatrix}}

\newcommand{\be}{\begin{equation}\begin{aligned}}
\newcommand{\ee}{\end{aligned}\end{equation}}

\def \QED{\hspace*{\fill}$\blacksquare$\medskip}

\newtheorem{definition}{Definition}
\newtheorem{theorem}{Theorem}

\title{\LARGE \bf
A novel model class for bowtie biological networks with universal classification properties
}

\author{Charles A. Johnson$^{1}$, Keon Ho (Daniel) Park and Enoch Yeung% <-this % stops a space
\thanks{*This work was funded in part by an NSF CAREER Award 2240176, the Army Young Investigator Program Award W911NF2010165 and the Institute of Collaborative Biotechnologies/Army Research Office grants W911NF19D0001, W911NF22F0005, W911NF190026, and
W911NF2320006. This work was also supported in part by
a subcontract awarded by the Pacific Northwest National
Laboratory for the Secure Biosystems Design Science Focus
Area “Persistence Control of Engineered Functions in Complex Soil Microbiomes” sponsored by the U.S. Department of Energy Office of Biological and Environmental Research.}% <-this % stops a space
\thanks{$^{1}$Charles A. Johnson ({\tt\small cajohnson@ucsb.edu}), Keon Ho (Daniel) Park ({\tt\small keonhopark@ucsb.edu}) and Enoch Yeung ({\tt\small eyeung@ucsb.edu})  are with Faculty of Mechanical Engineering
        UC Santa Barbara, Santa Barbara, USA
        }%
}

\begin{document}

\maketitle
\thispagestyle{empty}
\pagestyle{empty}

%%%%%%%%%%%%%%%%%%%%%%%%%%%%%%%%%%%%%%%%%%%%%%%%%%%%%%%%%%%%%%%%%%%%%%%%%%%%%%%%
\begin{abstract}

Cell sensory transcription networks are the intracellular computation structure that regulates and drives cellular activity. Activity in these networks determines the the cell's ability to adapt to changes in its environment. Resilient cells successfully identify (classify) and appropriately respond to environmental shifts. We present a model for identification and response to environmental changes in resilient bacteria. This model combines two known motifs in transcription networks: dense overlapping regulons (DORs) and single input modules (SIMs). DORs have the ability to perform cellular decision making and have a network structure similar to that of a shallow neural network, with a number of input transcription factors (TFs) mapping to a distinct set of genes. SIMs contain a master TF that simultaneously activates a number of target genes. Within most observed cell sensory transcription networks, the master transcription factor of SIMs are output genes of a DOR creating a fan-in-fan-out (bowtie) structure in the transcriptional network. We model this hybrid network motif (which we call the DOR2SIM motif) with a  superposition of modular nonlinear functions to describe protein signaling in the network and basic mass action kinetics to describe the other chemical reactions in this process. We analyze this model's biological feasibility and capacity to perform classification, the first step in adaptation. We provide sufficient conditions for models of the DOR2SIM motif to classify constant (environmental) inputs. These conditions suggest that generally low monomer degradation rates as well as low expression of source node genes at equilibrium in the DOR component enable classification.

\end{abstract}

%%%%%%%%%%%%%%%%%%%%%%%%%%%%%%%%%%%%%%%%%%%%%%%%%%%%%%%%%%%%%%%%%%%%%%%%%%%%%%%%
\section{Introduction}

To respond to changes in their environment, bacteria and other cells need to perform in-cell computation through networks of chemical reactions. The transcription and translation of messenger RNA into proteins interplays with environmental signals to identify and respond to new environments \cite{jacob1961genetic}. Successful cells will correctly classify their current environment and appropriately produce the proteins needed for that context \cite{lopez2008tuning}. Classification can be thought of as a control policy acting on a continuous signal space and reducing it to a discrete action from a finite or countably finite list of options. It is interesting because it is a black-box “analog-to-digital” or “continuous-to-discrete” transformation that biological cells have to make each day. We ask the question, ``how is biological classification represented computationally and structurally in the architecture and evolved design of living cells?'' Answering this question opens the door to further inquiry. What are the energy and metabolic benefits of such control and how can these evolved computation structures be applied to improve the efficiency of human-designed systems such as AI?

Dedicated biologists and systems theorists have identified common network motifs implemented in almost all cells \cite{shen2002network}. %In this paper we combine these motifs in a matter consistent with common biological patterns to propose a specific model for biological classification and response to environmental signals.
%
%Cell sensory transcription networks are the in-cell method of computation that regulates and drives activity within the cell. In these networks, special genes known as transcription factors, are activated by environmental signals. Their activation promotes or represses the production of messenger RNA tied to the expression of one or multiple genes. The expression of these genes then produces transcription factors and other proteins. In theory, the network structure of transcription networks is combinatorial. However, transcription networks in wild-type cells mostly exhibit a small number of network motifs, combined according to a small number of rules \cite{?}. 
%
Two of these motifs are \textit{dense overlapping regulons} and \textit{single input modules}. Dense overlapping regulons have the ability to perform cellular decision making and have a network structure akin to that of a shallow neural network, with a number of input transcription factors (input TFs) mapping to a distinct set of genes \cite{alon2019introduction}. Single input modules contain a master TF (typically auto regulated) that simultaneously activates a number of target genes. The activation of these target genes typically enables the sequential production of proteins which combine together to form an enzyme. The final product enzyme may have multiple functions, its primary function may be related to performance in a particular environment. Typically, its secondary function is to deactivate the master TF \cite{alon2019introduction}. 

Within most observed cell sensory transcription networks, dense overlapping regulons are not layered in cascade with other dense overlapping regulons. This means that intracellular classification is not made by a deep network. Additionally, the master TF of single input modules are an output gene of a dense overlapping regulon \cite{shen2002network}. This is an example of the bowtie motif in biological networks \cite{csete2004bow, friedlander2015evolution}.

We build this bowtie network by combining models of transcriptional networks from \cite{alon2019introduction} and \cite{del2015biomolecular} to hypothesize a model for bacterial decision making. Given this model, we analyze its biological feasibility and basic capability for classification. Part of our interest in this motif stems from its network structure. Because all ambiguous components of our proposed model's network structure are feedforward; results from linear network reconstruction suggest that this model's topology and parameters could be uniquely identifiable from time-series experimental data \cite{gonccalves2008necessary, materassi2010topological}.

In this paper, we study a class of biological networks that have the ability to convert continuous environmental signals into discrete control policies, namely that are able to perform a mode of biological classification. We introduce this class of models as the DOR2SIM model and provide de novo mathematical definitions and analyze their architectural properties. We consider a 11-state model that is constructed as the minimal set of signal TFs needed to build a distinguishable DOR motif and the minimal set of target genes to make a distinguishable SIM motif. We show that our 11-state DOR2SIM model is capable of classifying any constant input signal from the environment and thus, prove by inductive logic, that DOR2SIM models are a powerful class of models for universally classifying all positive constant signals. We then provide a sufficient model parameter condition to verify the ability of any DOR2SIM model to classify constant inputs based off of an experimentally measured or simulated equilibrium point of the system. This parameter condition suggests that the behavior of source nodes in the DOR motif at equilibrium leads to successful classification. As a preliminary, we define dynamic classification.

\begin{figure}
    \centering
    \includegraphics[width=0.95\linewidth]{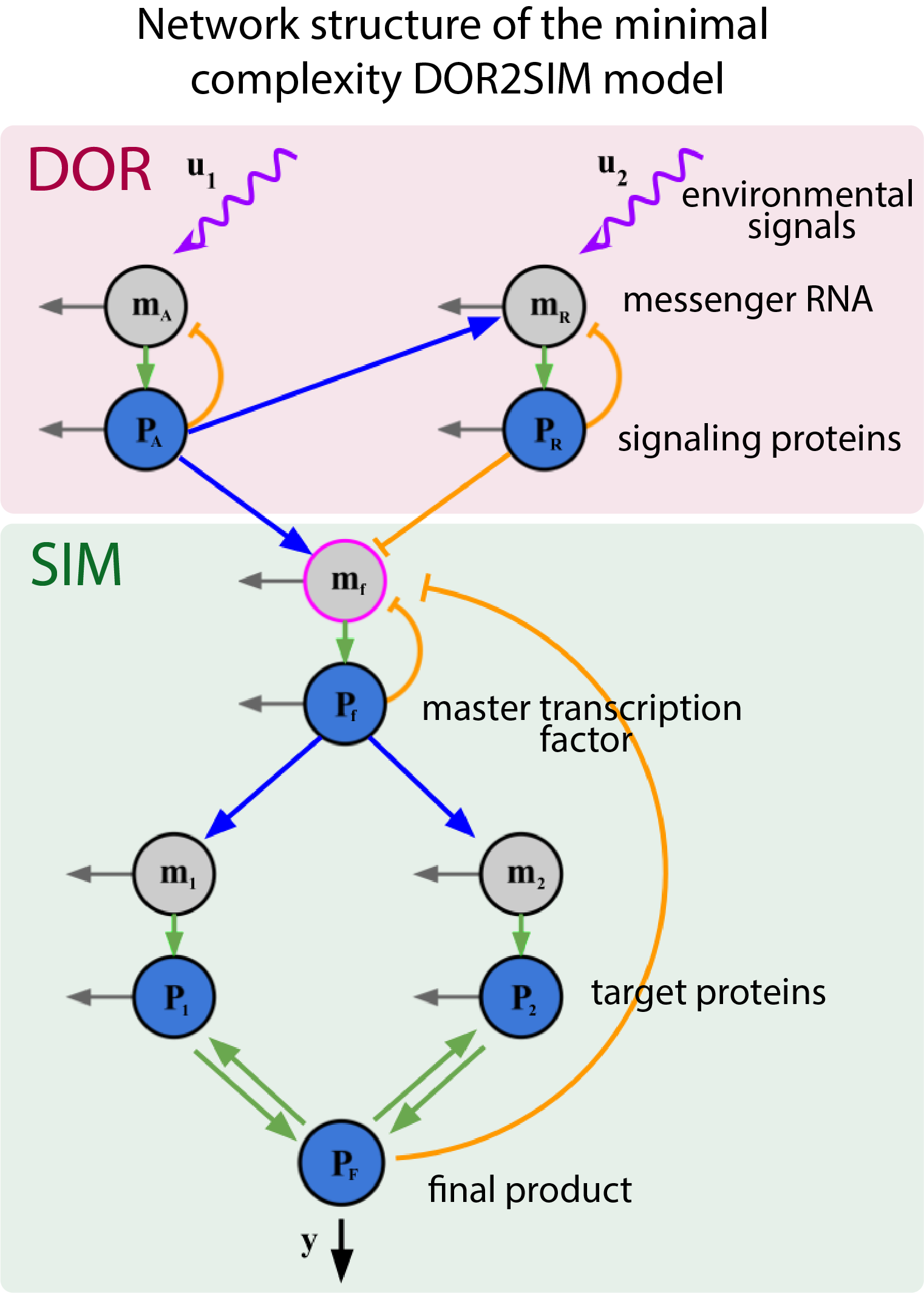}
    \caption{A diagram of the network structure of the DOR2SIM model. The dense overlapping regulon (DOR) and single input module (SIM) components are highlighted in red and green respectively. Nodes are labeled consistent with Equation (\ref{eq:toy_example}).}
    \label{fig:network}
\end{figure}

\section{Dynamic classification of inputs}

In this paper we consider the classification properties of a nonlinear dynamic system of the form 

\be \label{eq:dyn_sys}
\dot x = f(x, u),
\ee where $x\in\R^n, u\in\R^m$. We define classification as the act of sorting input to the system into distinct classes. These classes are distinguished from each other by the systems state variables. This information on the class of input should be stored in the system until the input changes, at which point it should reclassify the new input. This definition of the dynamic classification of inputs is given precisely below.

\begin{definition}\label{def:classification_sys}
The dynamic system in Equation (\ref{eq:dyn_sys}) is defined to classify a constant-valued input, $\bar u\in\R^m$ over the initial point set $S_{ic}\subset \R^n$, to the class set $S_c\subset \R^n$ if, for all initial points in the initial point set, $x_0\in S_{ic}$, we have that our constant valued input and the system dynamics evolved from our initial point, $\dot x = f(x, \bar u)$, $x(0)=x_0$, implies that there exists a threshold time, $T\in\R^+$, so that for any time greater than the threshold time the state is in the class set, $\exists\mbox{ } T>0$ so that for all $t>T$ $x(t)\in S_c$. See Figure \ref{fig:input_classification} for a visual schematic of this definition.
\end{definition}

\begin{figure}
    \centering
    \includegraphics[width=0.95\linewidth]{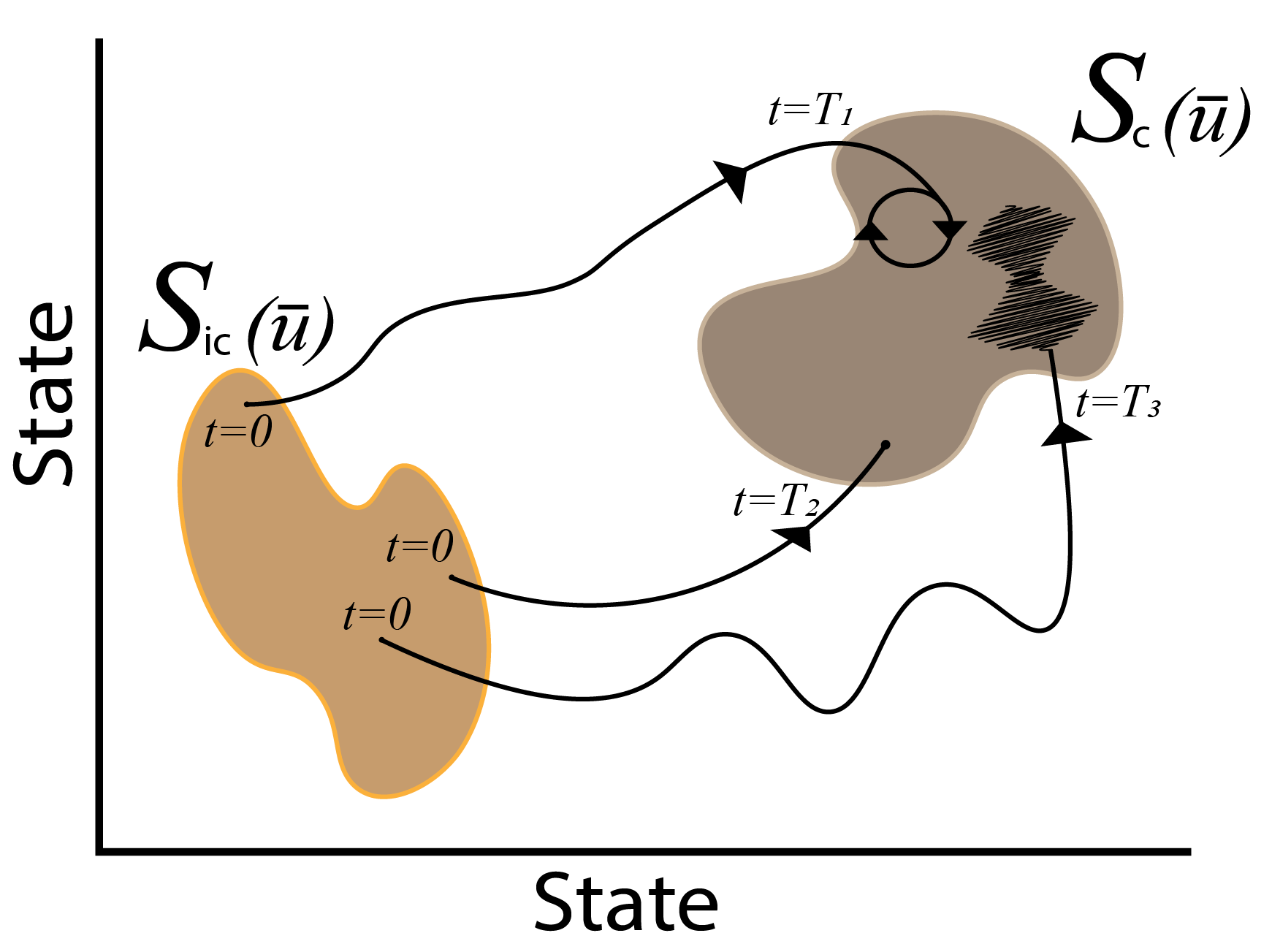}
    \caption{A visual schematic of Definition \ref{def:classification_sys}. Given a constant input, $\bar u$, trajectories that begin in $S_{ic}$ will reach and remain in $S_c$ in finite time.}
    \label{fig:input_classification}
\end{figure}

Cellular combinatorial decision making in transcriptional networks is performed in dense overlapping regulons (DORs) and then applied in (among other motifs) single input modules (SIMs). This means that classification models of DORs connecting to SIMs should be able to classify any constant-valued inputs to an appropriate class set. The definition of what set might be appropriate is open for discussion, however, in this paper ``appropriate'' sets will be simply connected sets. The most important feature in a dynamic system for classification is the ability to distinguish between class sets for different inputs. This means that, if a small perturbation is applied to the inputs, the intersection of the new class set with the class set of the unperturbed inputs should be the empty set.  We formalize this notion in the following definition of an input classifying system.

\begin{definition}
    The dynamic system in Equation (\ref{eq:dyn_sys}) is an input classifying system over an interval, $\mathcal{I}$, of $\R^m$ if it classifies all constant inputs in $\mathcal{I}$ to disjoint class sets. 
\end{definition}

Before we explore the classification potential of our model of the DOR2SIM network motif, we introduce the dynamics of the DOR and SIM models separately and then we combine them.

\section{Dense overlapping regulon model}

In a dense overlapping regulon (DOR), environmental signals influence the production of messenger RNA that codes to transcription factors (TFs) that will influence the downstream expression of genes that respond to the environmental signal.  This is the controller component of the DOR2SIM motif. In this model the dynamic states will be the concentrations of messenger RNA as well as the concentrations of the proteins that they code for. The dynamics can be given by a function, $f_{DOR}:\R^{2n}\times\R^n\rightarrow\R^{2n}$, and the outputs are the concentrations of the final TFs, $y_{DOR}:\R^{2n}\rightarrow\R^{n}.$ These functions are defined as follows 

\be
x_{DOR} &\triangleq [m_1, ..., m_n, P_1, ..., P_n]^T \\
\dot x_{DOR} &= f_{DOR}(x_{DOR}, u),\ee 
where:
\be 
\dot m_i &= F_{m_i}(P_i, P_{i+1}, ..., P_n) - \delta_{i}m_i + b_i u_i \\
\dot P_i &= k_{P_i}m_i - \delta_{n+i}P_i \\
y_{DOR} &= \BM P_1&P_2&\hdots&P_n \EM^T,
\ee 
 where each $\delta, b$ and $k$ is a constant parameter and each function $F_m$ is a special nonlinear function which describes the binding of the input proteins to the promoter site for the specified messenger RNA.

 So, what is the functional form of $F_{m_i}$, our nonlinear binding function? Well this depends on the type of promoter that we are assuming in our model. In this paper we assume that binding is competitive and only a single protein may bind at a time.
 
When the binding site only permits a single protein to bind at a time, the activation/repression dynamics in $F_{m_i}$ are given (as derived in Chapter 2 Section 2.3 of \cite{del2015biomolecular}) by:
\be\label{eq:competitive_binding}
F_{m_i}(P_i, P_{i+1},...,P_n) = \frac{\beta_{m_i}\sum_{j\in A}^{}(P_j/K_j)^{n_j}+\alpha_0}{1 + \sum_{j=i}^{n}(P_j/K_j)^{n_j}},
\ee where $A\subset \{i, i+1, ..., n\}$ is the set of all indices corresponding to activator proteins (the remaining indices correspond to repressor proteins), $\beta_{m_i}$ is a constant, $\alpha_0$ is a constant giving the rate of gene expression in the absence of activators on the promoter site and the parameters $K_j$ and $n_j$ are constants.

\section{Single input module model}
The single input module (SIM) is a motif in transcription networks where a single master TF, $P_M$, will activate the promoter sites for a number of other genes. These genes, in turn, produce proteins which combine with each other to form a single enzyme. In the DOR2SIM motif, this enzyme can perform multiple functions. Its primary purpose is typically to respond to the environmental signaling which activated the master transcription factor (master TF) in the first place. Additionally, it often down-regulates the master TF to turn off the system once the enzyme concentration is high enough.

The dynamics for the SIM will be given by the function $f_{SIM}:\R^{3N+1}\times\R^{n}\rightarrow\R^{3N+1}$. The output for the SIM is simply the final enzyme, so $y_{SIM}:\R^{3N+1}\rightarrow \R$. The system dynamics are given below as follows:

\be 
x_{SIM} &\triangleq [m_M, P_M, m_1, ..., m_N, P_1, ..., P_{2N-1}]^T \\
\dot x_{SIM} &= f_{SIM}(x_{SIM}, u),\ee 
where:
\be 
\dot m_M &= F_{M}(u_1, ..., u_n, P_M, P_{2N-1}) -\delta_{m_M}m_M\\
\dot P_M &= k_{M}m_M-\delta_{P_M}P_M\\
\dot m_{i\leq N} &= F_{m_i}(P_M) - \delta_{m_i}m_i \\
\dot P_1 &= k_1 m_1 + \delta_{P+{N+1}}P_{N+1} \\&\mbox{ }- k_{N+1}P_1P_2 -\delta_{P_1}P_1\\
\dot P_{2\leq i\leq N} &= k_i m_i + \delta_{P_{N+i-1}}P_{N+i-1} \\&\mbox{ } - k_{i+N-1}P_iP_{i-1} - \delta_{P_i}P_i\\
\dot P_{N+1} &= k_{N+1}P_1P_2 - \delta_{P_{N+1}}P_{N+1}\\
\dot P_{N+2\leq i} &= k_i P_{i-N+1}P_{i-1} - \delta_{P_i}P_i\\
y_{SIM} &= P_{2N-1},
\ee where each $\delta$ and $k$ is a constant corresponding to degradation rate and reaction rate respectively. The function $F_M$ and the functions $F_i$ are nonlinear functions describing the binding of input proteins to the promoter site for messenger RNA. 

Each function, $F_i$, takes the form of a single variable Hill activation function, $Hill$. In this paper we consider two types of Hill functions, Hill activation and Hill repressor functions. Hill activation functions are defined as follows: 
\be 
Hill_{act}(x) = \frac{\beta x^\mathbf{n}}{K^\mathbf{n} + x^\mathbf{n}}.
\ee Each $F_i$ is a Hill activation function. The function, $F_n$ in the DOR model is a Hill repressor function, defined as follows: 

\be
Hill_{rep}(x) = \frac{\beta K^\mathbf{n}}{K^\mathbf{n} + x^\mathbf{n}}.
\ee For both types of single-variable Hill functions the constants $\beta, K$ and $\mathbf{n}$ are function parameters.

On the other hand, the functional form of $F_M$ is more nuanced. Because of our underlying assumption that each promoter binding site only permits a single protein to bind at a time, $F_M$ will be modeled according to Equation (\ref{eq:competitive_binding}), much like the $F_m$ functions in the DOR model.

\section{The combined DOR2SIM model}

If we interconnect a DOR with a SIM, ignoring feedback due to the final product of the SIM influencing the signals that stimulate the expression of the input TFs, we have an open-loop model of the system. This model is as follows:

\be 
\dot x_{DOR} &= f_{DOR}(x_{DOR}, u)\\
\dot x_{SIM} &= f_{SIM}(x_{SIM}, y_{DOR}(x_{DOR}))\\
y &= y_{SIM}(x_{SIM}).
\ee

\subsection{The minimal complexity DOR2SIM model}

To study the dynamic classification properties of the DOR2SIM model class we will begin by proposing a minimal complexity instantiation of the DOR2SIM that retains all of the dynamic features of the model. This means that our model will include: \begin{enumerate}
    \item a feedforward connection between the signaling TFs
    \item negative auto regulation in both the signaling and master TFs as well as
    \item a higher order protein that is not coded for directly in the organism's DNA.
\end{enumerate} Complexity here is structural complexity and is measured as the number of signaling TFs in the DOR component and the number of target TFs in the SIM component. The minimal complexity DOR2SIM model will have only two signaling TFs in its DOR component and two target TFs in its SIM component. If we structurally simplify the model further by removing a signaling TF from the DOR or a target TF from the SIM, then the DOR and SIM motifs, respectively, would not be distinguishable from a cascade motif and so we no longer have a DOR2SIM model. Thus there are no further structural simplifications possible for this model.

However, we will operate under a number of non-structural simplifying assumptions. We will assume that 
\begin{enumerate}
    \item protein degradation rates for all monomer proteins are equal,
    \item translation rates for all monomer proteins are equal, 
    \item input signals are unweighted ($b_1=b_2=1$), 
    \item all promoters maximal production rate is 1 ($\beta=1$ in Hill functions),
    \item all promoters activate (or repress) at the promoter site after they form a dimer (which we model by setting the Hill coefficients to be 2 assuming high cooperativity \cite{santillan2008use}) and 
    \item half maximal repression for the source messenger RNA is reached when the concentration of the protein it codes for is at 1, this means that in Equation (\ref{eq:toy_example}): 
    \be 
        F_1(x_3)\triangleq \frac{\beta}{1 + x_3^2},
    \ee for some positive constant $\beta$.
\end{enumerate}

Ultimately, our minimal complexity DOR2SIM model is as follows:

\be \label{eq:toy_example}
\BM \dot m_a\\ \dot m_r \\ \dot P_a \\ \dot P_r \\ \dot m_f \\ \dot m_1 \\ \dot m_2 \\ \dot P_f \\ \dot P_1 \\ \dot P_2 \\ \dot P_F \EM \triangleq \BM \dot x_1\\ \dot x_2 \\ \dot x_3 \\ \dot x_4 \\ \dot x_5 \\ \dot x_6 \\ \dot x_7 \\ \dot x_8 \\ \dot x_9 \\ \dot x_{10} \\ \dot x_{11} \EM  = \BM F_1(x_3) - \delta x_1 + u_1 \\ F_2(x_3, x_4) - \delta x_2 + u_2 \\ k x_1 - \delta x_3 \\ k x_2 - \delta x_4 \\ F_3(x_3, x_4, x_8, x_{11})-\delta x_5 \\ F_4(x_8) - \delta x_6\\ F_5(x_8) \delta x_7 \\ kx_5 - \delta x_8 \\ kx_6 - \delta x_9 + \Delta x_{11} - \bar k x_9x_{10} \\ kx_7 - \delta x_{10} + \Delta x_{11} - \bar k x_9x_{10} \\ \bar k x_9x_{10} - \Delta x_{11} \EM, 
\ee where the parameters $\delta, \Delta, k, \bar k$ are each positive constants. A visual depiction of the network structure of this model is in Figure \ref{fig:network}.

\section{Classification properties of the DOR2SIM model} 
Our numerical exploration of this model indicates that, like other biological models arising from mass action kinetics, the DOR2SIM model is highly stable and that the equilibrium points vary with changes in input values, see Figure \ref{fig:demonstration}. These are encouraging results for the capacity of this model to perform classification. However, the nonlinear dynamics of these systems allow for these systems to potentially have multi-stability induced hysteresis. Multi-stability introduces a dependence between the ultimate state of the system on initial conditions. This dependence is problematic for classifying inputs since the same input could potentially, given the initial state of the system, map to distinct class sets. Thus, we delve into a mathematical analysis of the DOR2SIM model in order to clarify when it can model dynamic input classification in a biological context.

\begin{figure}
    \centering
    \includegraphics[width=1.08\linewidth]{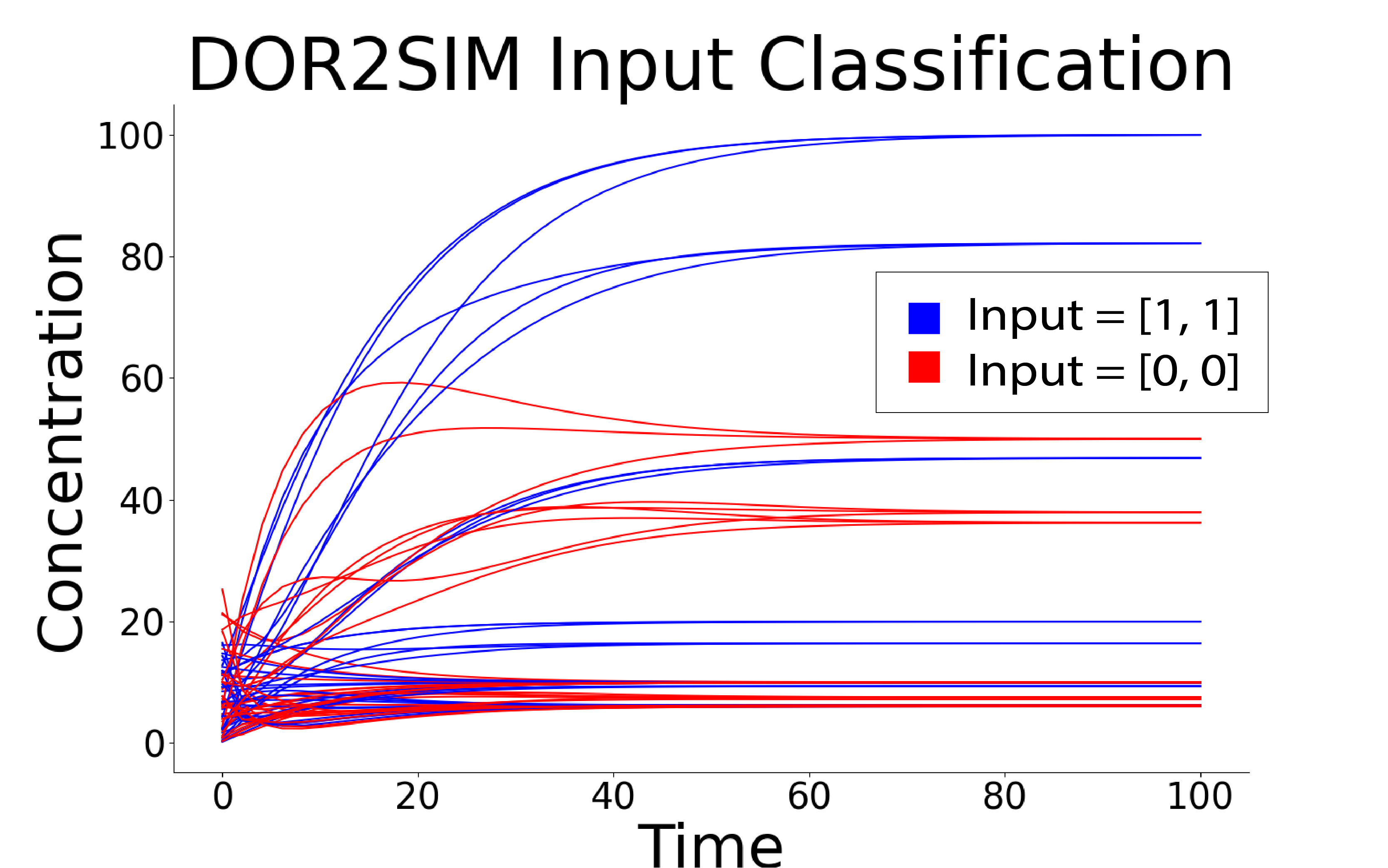}
    \caption{The evolution, in time, of multiple state trajectories from random initial conditions given distinct inputs for identical DOR2SIM models. Note that the equilibrium points depend on the input, rather than initial conditions demonstrating that this network classifies distinct inputs with with its state. The lines appear to overlap because we are plotting all 11 state trajectories vs time on the same axis. If it were an 11 dimensional graphic they would not overlap.}
    \label{fig:demonstration}
\end{figure}

\subsection{Defining biologically feasible initial conditions, inputs and class sets}

Recall from Definition \ref{def:classification_sys} that the DOR2SIM model will only be able to classify inputs from a set of initial conditions to a class set if any initial condition combined with the specified input results in a system trajectory that will reach the class set and remain in the class set in finite time. As such, our set of initial conditions will be the set of all biologically feasible initial conditions. Since each state variable in our system is the concentration of some compound, be it a protein or messenger RNA, biologically feasible initial conditions must be nonnegative and finite. So our biologically feasible initial conditions belong to the set: 

\be 
S_{ic} \triangleq \{x\in\R^{11}:x_i \geq 0, \forall i=1,2,...,11\}.
\ee 

Our inputs and class sets likewise must be nonnegative in order to be biologically feasible.

\subsection{Sufficient conditions for the 11-state DOR2SIM model to classify any biologically feasible input}

\begin{theorem}\label{thm:toy_model}
    Given the minimal complexity DOR2SIM model and assuming that every biologically feasible initial condition will lie in the bastion of attraction of some stable equilibrium point, if the monomer degradation rate satisfies, \be \label{eq:thm_toy_eq} \left(\frac{k^2}{3}u_1^2\right)^{\frac{1}{4}}<\delta<\frac{\sqrt{6}}{6}\approx 0.408,\ee then the model will classify the constant input $\bar u\triangleq [u_1, u_2]$ over the set of all biologically feasible initial conditions, $S_{ic}$, to the class set $S_c$ where \be S_c\triangleq \{x\in S_{ic}: x_1=m_a^{eq}, x_3=\frac{k}{\delta}m_a^{eq}\}, \ee for some $m_a^{eq}\in\R^+$. 

    Furthermore, the minimal complexity DOR2SIM model is an input classifying system over some sufficiently small intervals of the set of biologically feasible inputs that satisfies Equation (\ref{eq:thm_toy_eq}).
\end{theorem}

\begin{proof}
In this proof our objective will be to demonstrate that any equilibrium point of this system will belong to the class set $S_c$. We begin by noting that at equilibrium, \be
0 = kx_1^{eq} - \delta x_3^{eq}\implies x_3^{eq}=\frac{k}{\delta}x_1^{eq}=\frac{k}{\delta}m_a^{eq}.
\ee  We also have that \be 
0 = \frac{1}{1+c^2{x_1^{eq}}^2} - \delta x_1^{eq} + u_1,
\ee where we define $c\triangleq\frac{k}{\delta}$, to simplify notation. Multiplying both sides by $1+c^2{x_1^{eq}}^2$ and combining like terms we obtain the following equivalent statement: \be \label{eq:toy_eq_req}
0=-\delta c^2 {x_1^{eq}}^3 + u_1c^2{x_1^{eq}}^2-\delta{x_1^{eq}} + (1 + u_1).
\ee By Descartes rule of signs this polynomial will have 3 or 1 positive root and 0 negative roots. This means that there will be 0 or 2 complex roots. We will apply Sturm's Theorem to determine conditions where there will be exactly 1 positive root and 2 complex roots. To do this we build a Sturm sequence and verify the differences in the sign changes in the Sturm sequence at ${x_1^{eq}}=0$ and as ${x_1^{eq}}\rightarrow \pm\infty$. 

When we do so we have that the sequence signs are: \be 
(+, -, -\varsigma_1, \varsigma_2)&\mbox{ as } {x_1^{eq}}\rightarrow-\infty \\
(+, -, -, \varsigma_2)&\mbox{ when } {x_1^{eq}}=0 \\
(-, -, \varsigma_1, \varsigma_2)&\mbox{ as } {x_1^{eq}}\rightarrow \infty,
\ee where 
\be \label{eq:thm1_conditions}
\varsigma_1 &= sign(\frac{-2c^2u_1^2+6\delta^2}{9\delta}) = sign(3\delta^2-c^2u_1^2)\\
\varsigma_2 &= sign(\frac{P_{pos}(c, u_1, \delta)}{4c^2u_1^2-24c^2\delta^2u_1^2+36\delta^4}) \\&= sign(c^2u_1^2-6c^2u_1^2\delta^2+9\delta^4),
\ee where $P_{pos}$ is a $5^{th}$-order polynomial in $c,u_1$ and $\delta$ with all positive coefficients. In order to have zero real negative roots and one real positive root, Sturm's Theorem implies that $\varsigma_1$ and $\varsigma_2$ both be positive. For $\varsigma_1$ to be positive: \be 
0 < 3\delta^2 - c^2u_1^2  \implies   c^2u_1^2 < 3\delta^2 \implies \frac{k^2}{3}u_1^2 < \delta^4. 
\ee Since all the involved terms are positive we can maintain this inequality when we take the cubed root of both sides of the final inequality to get the left half of Equation (\ref{eq:thm_toy_eq}).

To conclude our proof, when $\varsigma_1$ is positive we will show that $\varsigma_2$ will also be positive if $\delta < \frac{\sqrt{6}}{6}$. First, when $\varsigma_1$ is positive we then have that, for some $\varepsilon>0$, that \be 
3\delta^2 = c^2u_1^2 + \varepsilon \implies 3\delta^2 - \varepsilon  = c^2u_1^2.
\ee  This means that we can simplify our term related to $\varsigma_2$ in Equation (\ref{eq:thm1_conditions}) as follows: \be
c^2u_1^2-6c^2u_1^2\delta^2+9\delta^4 &= (3\delta^2 - \varepsilon)-6(3\delta^2 - \varepsilon)\delta^2+9\delta^4 \\
&=-9\delta^4 + (3 + 6\varepsilon)\delta^2 - \varepsilon.
\ee This polynomial is positive (as can be verified using the quadratic formula on with the substitution $z \triangleq \delta^2$ to find zeros and then checking intervals) whenever: \be 
\frac{1 + 2\varepsilon}{6} - \frac{\sqrt{36\varepsilon^2+9}}{18} < \delta^2 < \frac{1 + 2\varepsilon}{6} + \frac{\sqrt{36\varepsilon^2+9}}{18}.
\ee At this point we can substitute \be 
\varepsilon = 3\delta^2 - c^2u_1^2
\ee into the above inequality and simplify to get: \be 
9\delta^4 - 6c^2u_1^2\delta^2+c^2u_1^2 > 0.
\ee We then recall that $c\triangleq\frac{k}{\delta}$, so we substitute and simplify again to get: \be 
9\delta^6 - 6k^2u_1^2\delta^2+k^2u_1^2 > 0.
\ee We then solve for $u_1^2$ and we get \be 
u_1^2 > \frac{-9\delta^6}{k^2(1 - 6\delta^2)}.
\ee  We then note that whenever $1-6\delta^2 > 0$, that $u_1$ satisfies this condition for all positive reaction rates, $k$. This finally brings us to the right side of Equation (\ref{eq:thm_toy_eq}) as follows, \be 
0< 1-6\delta^2 \implies 6\delta^2 < 1\implies \delta < \frac{1}{\sqrt{6}} = \frac{\sqrt{6}}{6}.
\ee The final claim of the theorem is validated by $x_1^{eq}$'s smooth dependence on $u_1$ in Equation (\ref{eq:toy_eq_req}), so the relationship between $x_1^{eq}$ and $u_1$ is bijective on a small enough interval that does not contain a local minimum or maximum. 
\end{proof}

The proof of Theorem \ref{thm:toy_model} demonstrates that the positive constant $m_a^{eq}$ is the concentration of the messenger RNA for the source node in the DOR's feedforward component of the DOR2SIM model. This same proof also demonstrates that $m_a^{eq}$ is a function of the system parameters and the input to the same source node, $u_1$. 

Theorem \ref{thm:toy_model} implies that as long as the TFs and messenger RNA of the system do not degrade too rapidly with respect to the timescale of the transcription dynamics that this model will successfully classify sufficiently small, constant inputs. Theorem \ref{thm:toy_model} is interesting in that monomer degradation rates are lower in cells that divide more slowly and previous research suggests a tradeoff between cell growth rate and environmental adaptation \cite{lopez2008tuning}. Theorem \ref{thm:toy_model} is also encouraging as the negative auto regulation built into this model promotes a quick response time \cite{alon2019introduction} and therefore the degradation parameter should be small by comparison. So this model should, under reasonable assumptions, successfully model dynamic classification for small enough inputs.

\section{Using data to determine if higher-order DOR2SIM models classify}

Numerical integration or experimental measurements of gene activity in a cell culture can be used to estimate a stable equilibrium point of the DOR2SIM model. However, the existence of such an equilibrium point is not sufficient to establish a DOR2SIM model's ability to classify the constant input used to find that equilibrium point. It is possible for biological systems to showcase hysteresis \cite{angeli2004detection}. This means that, although the measured or simulated trajectory (or trajectories) indicate a particular state that the system will map the input to, different initial conditions may be in the bastion of attraction for a different attracting set.  Given a stable equilibrium point we provide sufficient results to verify that the DOR2SIM system will map all biologically viable initial conditions to the generalization of the class set, $S_c$, defined in Theorem \ref{thm:toy_model}. As before, $S_c$, will be defined in terms of the concentration of the source messenger RNA in the feedforward component of the DOR within the DOR2SIM model.

In this section, we retain a few of the our previous modeling assumptions, specifically, we assume that \begin{enumerate}
    \item all promoters activate (or repress) at the promoter site after they form a dimer (so all Hill coefficients are set to 2) and 
    \item half maximal suppression for the source messenger RNA is reached when the concentration of the protein it codes for is at 1.
\end{enumerate}

\begin{theorem}\label{thm:full_model}
    Given a DOR2SIM system with a biologically feasible constant input, $\bar u = [u_1,...,u_n]$, resulting in a biologically feasible equilibrium point at $x^{eq}\in\R^{2n+3N+1}$ assume that every biologically feasible initial condition will lie in the bastion of attraction of some stable equilibrium point and that the following condition is satisfied: \be\label{eq:thm_full_eq}
    k_{P_n}^2 x_n^{eq}(\frac{b_nu_n}{\delta_{2n}^2} - x_n^{eq}\delta_n)^2  - 4\delta_n\delta_{2n}^2(\beta_n+b_nu_n) < 0,
    \ee where $x^{eq}_n$ is the concentration of the source messenger RNA in the feedforward dynamics of the DOR at the equilibrium point, $x^{eq}$, and where $k_{P_n}, \delta_n, \delta_{2n}, b_n, \beta_n$ are all constant parameters in the DOR2SIM model. Then the entire DOR2SIM system will classify all biologically feasible initial conditions to the class set: \be 
    S_c\triangleq\{x\in\R^{2n+3N+1}:x_n=x^{eq}_n, x_{2n}=\frac{k_{P_n}}{\delta_{2n}}x^{eq}_n\}.
    \ee 

    Furthermore, this DOR2SIM model is an input classifying system over some sufficiently small intervals of the set of biologically feasible inputs that satisfies Equation (\ref{eq:thm_full_eq}).
\end{theorem}

\begin{proof}
  We begin this proof in a form analogous to that of Theorem \ref{thm:toy_model} and find that at equilibrium: 
  \be 
0 = k_{P_n}x_{n} - \delta_{2n} x_{2n}\implies x_{2n}=\frac{k_{P_n}}{\delta_{2n}}x_n.
  \ee We need to show that the $n^{th}$ state at equilibrium must equal the $n^{th}$ state of our known equilibrium point, specifically, that $x_n=x_n^{eq}$. 

  Again, as in Theorem \ref{thm:toy_model}, we define a convenience variable, $c\triangleq\frac{k_{P_n}}{\delta_{2n}}$ and derive another required condition at equilibrium point from the state-evolution of $x_n$: \be \label{eq:poly_cond}
0 = -\delta_nc^2x_n^3 + b_nu_nc^2x_n^2-\delta_nx_n + (\beta_n+b_nu_n),
  \ee where $x_n$ is the concentration of the $n^{th}$ messenger RNA at equilibrium. Since all biologically feasible parameters and inputs to this system are positive Descartes rule of signs gives that this polynomial has 0 negative real solutions and 3 or 1 positive real solutions. We know that one positive real solution to this polynomial is $x_n=x_n^{eq}$. Therefore, to satisfy this equation with only one real solution, we must have that Equation (\ref{eq:poly_cond}) be writable as: \be
0  = (x_n - x_n^{eq})(\alpha x_n^2 + \beta x_n + \gamma),
  \ee with a negative discriminant: \be\label{eq:discriminant_cond} \beta^2 - 4\alpha\gamma < 0.\ee We then divide out $x_n - x_N^{eq}$ from Equation (\ref{eq:poly_cond}) and equate the coefficients to $\alpha, \beta$ and $\gamma$ to compute that: \be 
\alpha &= -\delta_nc^2\\
\beta &= b_n u_n c^2 - x_n^{eq}\delta_nc^2\\
\gamma &= -(\beta_n+b_n u_n) / x_n^{eq}. 
  \ee  Thus our negative discriminant condition, Equation (\ref{eq:discriminant_cond}), becomes Equation (\ref{eq:thm_full_eq}) and it implies that $x_n^{eq}$ is the only real value for $x_n$ which allows the system state $x$ to be at equilibrium. Similar to the proof of Theorem \ref{thm:toy_model}, we note the smooth dependence of $x_n^{eq}$ on $u_n$ in Equation (\ref{eq:poly_cond}) to justify the final statement of the theorem. 
\end{proof}

Theorem \ref{thm:full_model} shows that we have found an architecture (the DOR2SIM model) that universally classifies biologically feasible environments under mathematical conditions that coincide a relatively strong activation of the source node in the DOR component compared to its concentration of messenger RNA at equilibrium ($\beta_n+b_nu_n$ large compared to $x_n^{eq}$).  This suggests that the smaller the concentration of messenger RNA from the source node gene needed to keep the cell at equilibrium in one environment, the greater capacity that the cell has to classify based on smaller environmental signals and in spite of a less efficient promoter site. So this model suggests that the messenger RNA from source nodes in DOR2SIM motifs of highly adaptable bacteria will settle to low concentrations.

This theorem also provides a quick check to verify the ability of a particular DOR2SIM model to classify constant inputs of interest.  This check could be used in conjunction with numerical simulations or a sequence of experiments that independently tweak the input values to look for possible limits of the system's ability to classify inputs.

\section{Conclusion}

In order to survive, bacteria and other cells need to rely on decision making that is computed via networks of chemical reactions that interplay with the central dogma of biology: the transcription and translation of messenger RNA. This means that these cells need to be able to classify chemical and other signals from their environments through these dynamics. We have combined typical biological motifs to propose the DOR2SIM model, a specific model for biological classification and response to environmental signals. 

Our numerical analysis shows that this model is stable, consistent with experimentally measured transcriptomics and is capable of dynamic input classification. We have established a straightforward condition under which the minimal complexity DOR2SIM model is guaranteed to perform input classification. We have also established a sufficient condition to check the classification ability of an arbitrary DOR2SIM model receiving an arbitrary biologically consistent input. This condition is a simple computation that is made from the model parameters as well as an equilibrium point, which can be discovered through numerical simulation or from experimental data collection. The condition suggests that successful classification in these networks (and by extension, successful adaptation of cells to new environments) is tied to a maintaining a low concentration of messenger RNA from source nodes in the DOR at equilibrium. In future work we hope to leverage the network structure of the DOR2SIM model in order to identify this motif and dynamics in adaptable bacteria. We also hope to integrate the physical dynamics of DNA supercoiling into the model to understand the role that supercoiling plays in cell response to environmental stimulus.

\bibliography{root}

\end{document}